\documentclass[conference]{IEEEtran}
\usepackage[letterpaper, left=0.91in, right=0.91in, bottom=0.91in, top=0.75in]{geometry}
\newlength \figwidth
\usepackage{cite}
\ifCLASSINFOpdf
  \usepackage[pdftex]{graphicx}
	\usepackage{epstopdf}
  \graphicspath{{../Figures/}}
  \DeclareGraphicsExtensions{.eps}
\else
  \usepackage[dvips,draft]{graphicx}
\fi
\usepackage[cmex10]{amsmath}
\usepackage{amssymb}
\usepackage{pifont}
\usepackage{amsthm}
\usepackage{url}
\usepackage{dsfont}
\usepackage{tabulary}
\usepackage{multirow}
\usepackage{color}

\usepackage{color}
\usepackage{times}
\usepackage{verbatim}
\usepackage{floatflt}
\usepackage{enumerate}
\usepackage{array}
\usepackage{multicol,afterpage,wrapfig} 
\usepackage{tikz}
\usepackage[noend]{algpseudocode}
\makeatletter
\def\BState{\State\hskip-\ALG@thistlm}
\makeatother
\usepackage{amsmath,amssymb,eucal}
\usepackage[utf8]{inputenc}
\usepackage{epsfig}
\usepackage{exscale}
\usepackage{latexsym}
\usepackage{verbatim}
\usepackage{amsfonts}
\usepackage{multirow}
\usepackage{cite}
\usepackage{balance}
\usepackage{color}
\usepackage{transparent}
\usepackage{import}
\usepackage{mathtools}
\usepackage{tikz}
\usetikzlibrary{automata,arrows,positioning,calc}
\usepackage{bbm}
\usepackage[printonlyused,withpage]{acronym}
\usepackage{tabu}
\usepackage{longtable}
\usepackage{balance}
\usepackage{footnote}
\usepackage{tablefootnote}
\usepackage{enumitem}

\usepackage{adjustbox}
\usepackage{multirow}
\usepackage{pifont}
\def\BibTeX{{\rm B\kern-.05em{\sc i\kern-.025em b}\kern-.08em
    T\kern-.1667em\lower.7ex\hbox{E}\kern-.125emX}}
    
\renewcommand{\IEEEQED}{\IEEEQEDopen}

\newcommand*\xbar[1]{%
  \hbox{%
    \vbox{%
      \hrule height 0.5pt 
      \kern0.36ex
      \hbox{%
        \kern-0.12em
        \ensuremath{#1}%
        \kern-0.12em
      }%
    }%
  }%
}

\newfont{\bbb}{msbm10 scaled 500}

\newfont{\bb}{msbm10 scaled 1100}

\newcommand{\Rc}{{\cal R}}

\newcommand{\bx}{{\text{b}}}

\newcommand{\hx}{{\text{h}}}

\newcommand{\nx}{{\text{n}}}
\newcommand{\ox}{{\text{o}}}
\newcommand{\px}{{\text{p}}}

\newcommand{\sx}{{\text{s}}}

\renewcommand{\d}{\mathrm{d}}

\def\e{\text{e}}

\newcommand{\executeiffilenewer}[3]{%
\ifnum\pdfstrcmp{\pdffilemoddate{#1}}%
{\pdffilemoddate{#2}}>0%
{\immediate\write18{#3}}\fi%
}

\usetikzlibrary{arrows,calc}
\allowdisplaybreaks

\IEEEoverridecommandlockouts
\begin{document}
\pagenumbering{gobble}

\def \pcov{\mathcal{P}_\mathrm{cov}}
\def \pcovl{\mathcal{P}_\mathrm{cov}^\mathrm{L}}
\def \pcovn{\mathcal{P}_\mathrm{cov}^\mathrm{N}}
\def \pcovu{\mathcal{P}_\mathrm{cov}^\upsilon}
\def \pcovuav{\mathcal{P}_\mathrm{u}}
\def \pcovc{\mathcal{P}_\mathrm{g}}
\def \sinr{\mathsf{SINR}}
\def \du{d_\mathrm{u}}
\def \dc{d_\mathrm{g}}
\def \db{d_\mathrm{b}}
\def \ru{r_\mathrm{u}}
\def \Ru{R_\mathrm{u}}
\def \rc{r_\mathrm{g}}
\def \rb{r_\mathrm{b}}
\def \al{\alpha_\mathrm{L}}
\def \an{\alpha_\mathrm{N}}
\def \au{\alpha_\upsilon}
\def \bl{\beta_\mathrm{L}}
\def \bn{\beta_\mathrm{N}}
\def \bu{\beta_\upsilon}
\def \bx{\beta_\xi}
\def \pu{\mathrm{P_u}}
\def \pum{\mathrm{P_u^M}}
\def \pcm{\mathrm{P_g^M}}
\def \pc{\mathrm{P_g}}
\def \pcl{\mathrm{P_g^L}}
\def \pcn{\mathrm{P_g^N}}
\def \hu{\mathrm{h_u}}
\def \htx{h_\mathrm{t}}
\def \hrx{h_\mathrm{r}}
\def \hb{\mathrm{h_b}}
\def \hc{\mathrm{h_g}}
\def \pl{\mathcal{P}_\mathrm{L}}
\def \pn{\mathcal{P}_\mathrm{N}}
\def \pup{\mathcal{P}_\upsilon}
\def \px{\mathcal{P}_\xi}
\def \pluu{\mathcal{P}_\mathrm{L}^\mathrm{uu}}
\def \ml{m_\mathrm{L}}
\def \mn{m_\mathrm{N}}
\def \m{m_\upsilon}
\def \muu{\mathrm{m_{uu}}}
\def \muul{\mathrm{m_{uu}^L}}
\def \muun{\mathrm{m_{uu}^N}}
\def \mcb{\mathrm{m_{gb}}}
\def \mcbl{\mathrm{m_{gb}^L}}
\def \mcbn{\mathrm{m_{gb}^N}}
\def \mub{\mathrm{m_{ub}}}
\def \mubl{\mathrm{m_{ub}^L}}
\def \mubn{\mathrm{m_{ub}^N}}
\def \mcul{\mathrm{m_{gu}^L}}
\def \mcun{\mathrm{m_{gu}^N}}
\def \nl{n_\mathrm{L}}
\def \nn{n_\mathrm{N}}
\def \nx{n_\xi}
\def \n{n_\upsilon}
\def \sl{\psi_\mathrm{L}}
\def \sn{\psi_\mathrm{N}}
\def \su{\psi_\upsilon}
\def \sx{\psi_\xi}
\def \ol{\omega_\mathrm{L}}
\def \on{\omega_\mathrm{N}}
\def \ou{\omega_\upsilon}
\def \ox{\omega_\xi}
\def \il{I_\mathrm{L}}
\def \iN{I_\mathrm{N}}
\def \yl{y_\mathrm{L}}
\def \yn{y_\mathrm{N}}
\def \yu{y_\upsilon}
\def \lapi{\mathcal{L}_I}
\def \lapil{\mathcal{L}_{I_\mathrm{L}}}
\def \lapin{\mathcal{L}_{I_\mathrm{N}}}
\def \lapix{\mathcal{L}_{I_\xi}}
\def \lapic{\mathcal{L}_{I_\mathrm{g}}}
\def \lapiu{\mathcal{L}_{I_\mathrm{u}}}
\def \Ul{\Upsilon_\mathrm{L}}
\def \Un{\Upsilon_\mathrm{N}}
\def \Ux{\Upsilon_\xi}
\def \Uxu{\Upsilon_\xi^{\,\upsilon}}
\def \Unu{\Upsilon_\mathrm{N}^\upsilon}
\def \frl{f_{R_\mathrm{b}}^\mathrm{L}}
\def \frn{f_{R_\mathrm{b}}^\mathrm{N}}
\def \frx{f_{R_\mathrm{b}}^{\,\xi}}
\def \laml{\lambda_\mathrm{L}}
\def \lamn{\lambda_\mathrm{N}}

\def \t{\mathrm{T}}
\def \i{I}
\def \ic{I_\mathrm{g}}
\def \icl{I_\mathrm{g}^\mathrm{L}}
\def \icn{I_\mathrm{g}^\mathrm{N}}
\def \iu{I_\mathrm{u}}
\def \iul{I_\mathrm{u}^\mathrm{L}}
\def \iun{I_\mathrm{u}^\mathrm{N}}
\def \ful{\Phi_\mathrm{u}^\mathrm{L}}
\def \fun{\Phi_\mathrm{u}^\mathrm{N}}
\def \fu{\Phi_\mathrm{u}}
\def \fcl{\Phi_\mathrm{g}^\mathrm{L}}
\def \fcN{\Phi_\mathrm{g}^\mathrm{N}}
\def \fc{\Phi_\mathrm{g}}

\def \gb{\mathrm{G_b}}
\def \gu{\mathrm{G_{u}}}
\def \gb{\mathrm{G_b}}
\def \tt{\theta_\mathrm{t}}
\def \tb{\theta_\mathrm{b}}

\def \lamci{\hat{\lambda}_\mathrm{g}}
\def \lamb{\lambda_\mathrm{b}}
\def \lamu{\lambda_\mathrm{u}}
\def \lamg{\lambda_\mathrm{g}}
\def \lamul{\lambda_\mathrm{u}^\mathrm{L}}
\def \lamun{\lambda_\mathrm{u}^\mathrm{N}}
\def \lamcl{\lambda_\mathrm{g}^\mathrm{L}}

\def \iccl{I_\mathrm{gg}^\mathrm{L}}
\def \iccn{I_\mathrm{gg}^\mathrm{N}}
\def \icc{I_\mathrm{gg}}
\def \icu{I_\mathrm{gu}}
\def \iuc{I_\mathrm{ug}}
\def \iuu{I_\mathrm{uu}}
\def \iuul{I_\mathrm{uu}^\mathrm{L}}
\def \iuun{I_\mathrm{uu}^\mathrm{N}}
\def \icul{I_\mathrm{gu}^\mathrm{L}}
\def \icun{I_\mathrm{gu}^\mathrm{N}}
\def \iucl{I_\mathrm{ug}^\mathrm{L}}
\def \iucn{I_\mathrm{ug}^\mathrm{N}}
\def \lapiccl{\mathcal{L}_{\iccl}}
\def \lapiccn{\mathcal{L}_{\iccn}}
\def \lapicul{\mathcal{L}_{\icul}}
\def \lapicun{\mathcal{L}_{\icun}}
\def \lapicu{\mathcal{L}_{\icu}}
\def \lapiuc{\mathcal{L}_{\iuc}}
\def \lapicc{\mathcal{L}_{\icc}}
\def \lapiucl{\mathcal{L}_{\iucl}}
\def \lapiucn{\mathcal{L}_{\iucn}}
\def \lapiuu{\mathcal{L}_{\iuu}}
\def \lapiuul{\mathcal{L}_{\iuul}}
\def \lapiuun{\mathcal{L}_{\iuun}}
\def \pcb{\mathsf{p}_\mathrm{gb}}
\def \pub{\mathsf{p}_\mathrm{ub}}
\def \pcu{\mathsf{p}_\mathrm{gu}}
\def \pruu{\mathsf{p}_\mathrm{uu}}
\def \prcb{\mathsf{p}_\mathrm{gb}}
\def \fici{\hat{\Phi}_c}
\def \tl{\tau_\mathrm{L}}
\def \tn{\tau_\mathrm{N}}
\def \acbl{\alpha_{\mathrm{gb}}^\mathrm{L}}
\def \acbn{\alpha_{\mathrm{gb}}^\mathrm{N}}
\def \acul{\alpha_{\mathrm{gu}}^\mathrm{L}}
\def \acun{\alpha_{\mathrm{gu}}^\mathrm{N}}
\def \acb{\alpha_{\mathrm{gb}}}
\def \alcb{\alpha_{\mathrm{gb}}^\mathrm{L}}
\def \ancb{\alpha_{\mathrm{gb}}^\mathrm{N}}
\def \aubl{\alpha_{\mathrm{ub}}^\mathrm{L}}
\def \aubn{\alpha_{\mathrm{ub}}^\mathrm{N}}
\def \auu{\alpha_{\mathrm{uu}}}
\def \auul{\alpha_{\mathrm{uu}}^\mathrm{L}}
\def \auun{\alpha_{\mathrm{uu}}^\mathrm{N}}
\def \zl{\zeta_\mathrm{L}}
\def \zn{\zeta_\mathrm{N}}
\def \ec{\epsilon_\mathrm{g}}
\def \eu{\epsilon_\mathrm{u}}
\def \subl{\psi_\mathrm{ub}^\mathrm{L}}
\def \scbl{\psi_\mathrm{gb}^\mathrm{L}}
\def \scbn{\psi_\mathrm{gb}^\mathrm{N}}
\def \scul{\psi_\mathrm{gu}^\mathrm{L}}
\def \suu{\psi_\mathrm{uu}}
\def \suul{\psi_\mathrm{uu}^\mathrm{L}}
\def \suun{\psi_\mathrm{uu}^\mathrm{N}}
\def \zubl{\zeta_\mathrm{ub}^\mathrm{L}}
\def \zubn{\zeta_\mathrm{ub}^\mathrm{N}}
\def \zcul{\zeta_\mathrm{cu}^\mathrm{L}}
\def \zuu{\zeta_\mathrm{uu}}
\def \zuul{\zeta_\mathrm{uu}^\mathrm{L}}
\def \zuun{\zeta_\mathrm{uu}^\mathrm{N}}
\def \zcb{\zeta_\mathrm{gb}}
\def \zcbl{\zeta_\mathrm{gb}^\mathrm{L}}
\def \zcbn{\zeta_\mathrm{gb}^\mathrm{N}}
\def \pur{\rho_\mathrm{u}}
\def \pcr{\rho_\mathrm{g}}
\def \guu{\mathrm{g_{uu}}}
\def \gcu{\mathrm{g_{gu}}}
\def \gub{\mathrm{g_{ub}}}
\def \gubi{\mathrm{g_{ub}^{(i)}}}
\def \gcb{\mathrm{g_{gb}}}
\def \scb{\psi_\mathrm{gb}}
\def \hub{\mathrm{h_{ub}}}
\def \hcb{\mathrm{h_{gb}}}
\def \hcu{\mathrm{h_{gu}}}
\def \dub{d_\mathrm{{ub}}}
\def \dcb{d_\mathrm{{gb}}}
\def \dcu{d_\mathrm{{gu}}}
\def \publ{\Psi_\mathrm{ub}^\mathrm{L}}
\def \pcul{\Psi_\mathrm{gu}^\mathrm{L}}
\def \pcb{\Psi_\mathrm{gb}}
\def \pcbl{\Psi_\mathrm{gb}^\mathrm{L}}
\def \pcbn{\Psi_\mathrm{gb}^\mathrm{N}}
\def \puu{\Psi_\mathrm{uu}}
\def \puul{\Psi_\mathrm{uu}^\mathrm{L}}
\def \puun{\Psi_\mathrm{uu}^\mathrm{N}}
\def \buu{\beta_\mathrm{uu}}
\def \bcb{\beta_\mathrm{gb}}
\def \bubl{\beta_\mathrm{ub}^\mathrm{L}}

\def \hb{\mathrm{h}_{\mathrm{b}}}
\def \hu{\mathrm{h}_{\mathrm{u}}}
\def \hg{\mathrm{h}_{\mathrm{g}}}
\def \hx{\mathrm{h}_{\mathrm{x}}}
\def \hy{\mathrm{h}_{\mathrm{y}}}
\def \hxy{\mathrm{h}_{\mathrm{xy}}}
\def \x{\mathrm{x}}
\def \y{\mathrm{y}}
\def \zetaxy{\zeta_{\mathrm{xy}}}
\def \tauxy{\tau_{\mathrm{xy}}}
\def \tauoxy{\tau_{0,\mathrm{xy}}}
\def \gxy{g_{\mathrm{xy}}}
\def \psixy{\psi_{\mathrm{xy}}}
\def \psixyL{\psixy^{\mathrm{L}}}
\def \psixyN{\psixy^{\mathrm{N}}}
\def \L{\mathrm{L}}
\def \N{\mathrm{N}}
\def \a{\mathrm{a}}
\def \x{\mathrm{x}}
\def \y{\mathrm{y}}
\def \dxy{d_\mathrm{xy}}
\def \rxy{r_\mathrm{xy}}
\def \alphaxy{\alpha_\mathrm{xy}}
\def \alphaxyL{\alphaxy^\mathrm{L}}
\def \alphaxyN{\alphaxy^\mathrm{N}}
\def \Px{P_{\mathrm{x}}}
\def \Pxmax{P_{\mathrm{x}}^{\textrm{max}}}
\def \rhox{\rho_{\mathrm{x}}}
\def \epsx{\epsilon_{\mathrm{x}}}
\def \mxy{\mathrm{m}_{\mathrm{xy}}}
\def \mxyL{\mxy^{\mathrm{L}}}
\def \mxyN{\mxy^{\mathrm{N}}}
\def \Ixy{I_{\mathrm{xy}}}
\def \Phib{\Phi_{\mathrm{b}}}
\def \Phig{\Phi_{\mathrm{g}}}
\def \Phiu{\Phi_{\mathrm{u}}}
\def \Phihatg{\hat{\Phi}_{\mathrm{g}}}
\def \Phihatgb{\hat{\Phi}_{\mathrm{gb}}}
\def \Phihatgu{\hat{\Phi}_{\mathrm{gu}}}


\def \pcovuav{\mathcal{C}_\mathrm{u}}
\def \pcovc{\mathcal{C}_\mathrm{g}}
\def \rMuu{\mathrm{r_M}}
\def \ru{\mathrm{r_u}}
\def \Ru{R_\mathrm{u}}
\def \u{\mathrm{u}}
\def \c{\mathrm{g}}
\def \q{\mathrm{q}}
\def \r{\mathrm{r}}
\def \m{\mathrm{m}}
\def \prxy{\mathsf{p}_\mathrm{xy}}
\def \prxyxi{\mathsf{p}_\mathrm{xy}^\xi}
\def \pl{\mathcal{P}^\mathrm{L}}
\def \pn{\mathcal{P}^\mathrm{N}}
\def \rc{\mathrm{r_g}}
\def \d{\mathrm{d}}
\def \L{\mathrm{L}}
\def \N{\mathrm{N}}
\def \pu{P_\mathrm{u}}
\def \pumax{\mathrm{P_u^{max}}}
\def \lapixyxi{\mathcal{L}_{I_\mathrm{xy}^\xi}}
\def \px{P_\mathrm{x}}
\def \sy{\mathrm{s_y}}
\def \sixy{\psi_\mathrm{xy}}
\def \zetxy{\zeta_\mathrm{xy}}
\def \prxi{\mathcal{P}^\xi}
\def \s{\mathrm{s}}
\def \axy{\alpha_\mathrm{xy}}
\def \bxy{\beta_\mathrm{xy}}
\def \hxy{\mathrm{h_{xy}}}
\def \Rc{R_\mathrm{g}}
\def \pc{P_\mathrm{g}}
\def \lamc{\lambda_\c}
\def \D{\mathrm{D}}
\def \sigmau{\sigma_\mathrm{u}}
\def \g{\mathrm{g}}


\newtheorem{Theorem}{\bf Theorem}
\newtheorem{Corollary}{\bf Corollary}
\newtheorem{Remark}{\bf Remark}
\newtheorem{Lemma}{\bf Lemma}
\newtheorem{Proposition}{\bf Proposition}
\newtheorem{Assumption}{\bf Assumption}
\newtheorem{Approximation}{\bf Approximation}
\newtheorem{Definition}{\bf Definition}

\title{Cellular UAV-to-UAV Communications}
\author{\IEEEauthorblockN{M.~Mahdi~Azari$^{\star}$, Giovanni Geraci$^{\diamond}$, Adrian Garcia-Rodriguez$^{\dag}$, and Sofie~Pollin$^{\star}$}\\ \vspace{-0.3cm}
\normalsize\IEEEauthorblockA{\emph{$^{\star}$KU Leuven, Belgium \enspace $^{\diamond}$Universitat Pompeu Fabra, Spain \enspace $^{\dag}$Nokia Bell Labs, Ireland}}
\thanks{
The work of G.~Geraci was supported by the Postdoctoral Junior Leader Fellowship Programme from ``la Caixa" Banking Foundation.}
}
\maketitle
\thispagestyle{empty}
\begin{abstract}
Reliable and direct communication between unmanned aerial vehicles (UAVs) could facilitate autonomous flight, collision avoidance, and cooperation in UAV swarms. In this paper, we consider UAV-to-UAV (U2U) communications underlaying a cellular network, where UAV transmit-receive pairs share the same spectrum with the uplink (UL) of cellular ground users (GUEs).
We evaluate the performance of this setup through an analytical framework that embraces realistic height-dependent channel models, antenna patterns, and practical power control mechanisms. Our results demonstrate that, although the presence of U2U communications may worsen the performance of the GUEs, such effect is limited as base stations receive UAV interference through their antenna sidelobes. Moreover, we illustrate that the quality of all links degrades as the UAV height increases---due to a larger number of line-of-sight interferers---, and how the performance of the U2U links can be traded off against that of the GUEs by varying the UAV power control policy. 
\end{abstract}

\IEEEpeerreviewmaketitle
\section{Introduction}
\label{sec:Intro}

The telecommunications industry and academia have long agreed on the social benefits that can be brought by having cellular-connected unmanned aerial vehicles (UAVs) \cite{Qualcomm2017,GarGerLop2018,MozSaaBen2018}. These include facilitating search-and-rescue missions, acting as mobile small cells for providing coverage and capacity enhancements \cite{azari2017ultra}, and even automating logistics in indoor warehouses \cite{Cbinsights:19}. From a business standpoint, mobile network operators are chasing new revenue opportunities by offering cellular coverage to a heterogeneous population of terrestrial and aerial users \cite{Ericsson:18,YanLinLi18}. A certain consensus has been reached---both at 3GPP meetings and in the classroom---on the fact that present-day networks will be able to support cellular-connected UAVs up to a certain extent \cite{azari2018cellular,fotouhi2018survey,LopDinLi2018GC,NguAmoWig2018,azari2017coexistence}. Besides, recent studies have shown that 5G-and-beyond hardware and software upgrades will be required by both mobile operators and UAV manufacturers to target large populations of UAVs flying at high altitudes \cite{GerGarGal2018,DanGarGerICC2019}.

However, important use-cases exist where direct communication between UAVs, bypassing ground network infrastructure, would be a key enabler. These include autonomous flight of UAV swarms, collision avoidance, and UAV-to-UAV relaying, data transfer, and gathering. Similarly to ground device-to-device (D2D) communications \cite{LinAndGho2014,ChuCotDhi2017}, UAV-to-UAV (U2U) communications may also bring benefits in terms of spectral and energy efficiencies, extended cellular coverage, and reduced backhaul demands \cite{zeng2019accessing,ZhaZhaDi19}.

In this article, we investigate U2U communications underlaying a cellular network. In such a setup, UAV-to-UAV transmit-receive pairs share the same spectrum with the uplink (UL) of cellular ground users (GUEs). Through stochastic geometry tools, we explicitly characterize the performance of both U2U and GUE UL, as well as their interplay. To the best of our knowledge, this work is the first one to do so by accounting for: (i) a realistic propagation channel model that depends on the UAV altitude, (ii) the impact of a practical base station (BS) antenna pattern, and (iii) a fractional power control policy implemented by all nodes. Our takeaways can be summarized as follows:
\begin{itemize}[leftmargin=*]
\item The presence of U2U links may degrade the GUE UL. However, such performance loss is not dramatic, since BSs perceive interfering UAVs through their antenna sidelobes, and UAVs can generally transmit at low power thanks to the favorable U2U channel conditions.
\item The performance of both U2U and GUE UL links degrades as UAVs fly higher. This is due to an increased probability of line-of-sight (LoS)---and hence interference---on all UAV-to-UAV, GUE-to-UAV, and UAV-to-BS interfering links. This negative effect outweighs the benefits brought by having larger GUE-to-UAV and UAV-to-BS distances.
\item The UAV power control policy has a significant impact on all links. A tradeoff exists between the performance of U2U and UL GUE communications, whereby increasing the UAV transmission power improves the former at the expense of the latter.
\item Smaller U2U distances can improve the performance of both U2U and GUE UL. Indeed, owed to a better U2U path loss, UAVs may employ a smaller transmission power and therefore reduce the interference they cause to other U2U links and to GUEs.
\end{itemize}


\section{System Model}
\label{sec:System_Model}

In this section, we introduce the network topology, channel model, and power control mechanisms considered throughout the paper. 
The main parameters used in our study are given in Table~\ref{table:parameters}.

\subsection{Network Topology and Spectrum Sharing Mechanism}

\subsubsection*{Ground cellular network}
We consider the UL of a traditional ground cellular network as depicted in Fig.~\ref{U2U_SystemModel}, where BSs are uniformly distributed as a Poisson point process (PPP) $\Phib \in \mathbb{R}^2$ with density $\lamb$. All BSs are deployed at a height $\hb$, and communicate with their respective sets of connected GUEs. Assuming that the number of GUEs is sufficiently large when compared to that of the BSs, the active GUEs on each time-frequency physical resource block (PRB) form an independent Poisson point process $\Phig \in \mathbb{R}^2$ with density $\lamg=\lamb$ \cite{ChuCotDhi2017}. We further consider that GUEs associate to their closest BS, which generally also provides the largest reference signal received power (RSRP)\footnote{A GUE may connect to a BS $b$ other than the closest one $a$ if its link is in LoS with $b$ and not with $a$. However, since the probability of LoS decreases with the distance, such event is unlikely to occur \cite{3GPP36777}.}. Therefore, the 2-D distance between a GUE and its associated BS follows a Rayleigh distribution with a scale parameter given by $\sigma_\c = 1/\sqrt{2\pi \lamc}$. When focusing on a typical BS serving its associated GUE, the interfering GUEs form a non-homogeneous PPP with density $\lamci(r) = \lamb(1-e^{-\lamb \pi r^2})$, where $r$ is the 2-D distance between the interfering GUE and the typical BS \cite{ChuCotDhi2017,SinZhaAnd:15,YanGerQue:16}.

\subsubsection*{UAV-to-UAV communications}
As illustrated in Fig.~\ref{U2U_SystemModel}, in this work we also consider that U2U transmit-receive pairs reuse the cellular GUE UL spectrum. We assume that U2U transmitters form a PPP $\Phiu$ with intensity $\lamu$, and that each U2U receiver is randomly and independently placed around its associated transmitter with distance $\Ru$ distributed as $f_{\Ru}(\ru)$.
While our analysis holds for any transmit/receive UAV height, in the following we assume all UAVs to be located at the same height $\hu$, to evaluate the impact of such parameter.

\subsubsection*{Spectrum sharing} We consider an underlay in-band approach for resource sharing between GUE UL and U2U \cite{LinAndGho2014}, where each PRB may be used by both link types. This results in four types of links: (i) GUE-to-BS communication and/or interfering links, (ii) UAV-to-BS interfering links, (iii) UAV-to-UAV communication and/or interfering links, and (iv) GUE-to-UAV interfering links.

\subsection{Propagation Channel and Power Control}

We assume that any radio link between nodes $\x$ and $\y$ is affected by large-scale fading $\zetaxy$ (comprising path loss $\tauxy$ and antenna gain $\gxy$) and small-scale fading $\psixy$.

\subsubsection*{Probability of LoS}
We consider that links experience line-of-sight (LoS) and non-LoS (NLoS) propagation conditions with probabilities $\prxy^\L$ and $\prxy^\N$, respectively. In what follows, the superscripts $\mathrm{L}$ and $\mathrm{N}$ will denote system parameters under LoS and NLoS conditions, respectively. In our analysis we assume that $\prxy^\L$ is or can be approximated by a step function, i.e., $\prxy^\L$ is constant for an interval $[\r_i,\r_{i+1}]$, where $i = 1,2,\ldots$ and $0=\r_1 < \r_2 <\ldots$. 

\subsubsection*{Path loss}
The distance-dependent path loss between two nodes $\mathrm{x}$ and $\mathrm{y}$ is given by
\begin{equation}
\tauxy = \hat{\tau}_\mathrm{xy} \, \dxy^{\,\alphaxy},
\end{equation}
where $\hat{\tau}_\mathrm{xy}$ denotes the reference path loss, $\alphaxy$ is the path loss exponent, and $\dxy = \sqrt{\rxy^2 + \hxy^2}$, $\rxy$, and $\hxy = \hx - \hy$ represent the 3-D distance, 2-D distance, and height difference between $\mathrm{x}$ and $\mathrm{y}$, respectively. Table~\ref{table:parameters} lists the path loss parameters employed in our study, which depend on the nature of $\x$ and $\y$. In the sequel, we employ the subscripts $\{\u,\c,\mathrm{b}\}$ to denote UAV, GUE, and BS nodes, respectively.

\begin{figure}
	\centering
	\includegraphics[width=0.95\figwidth]{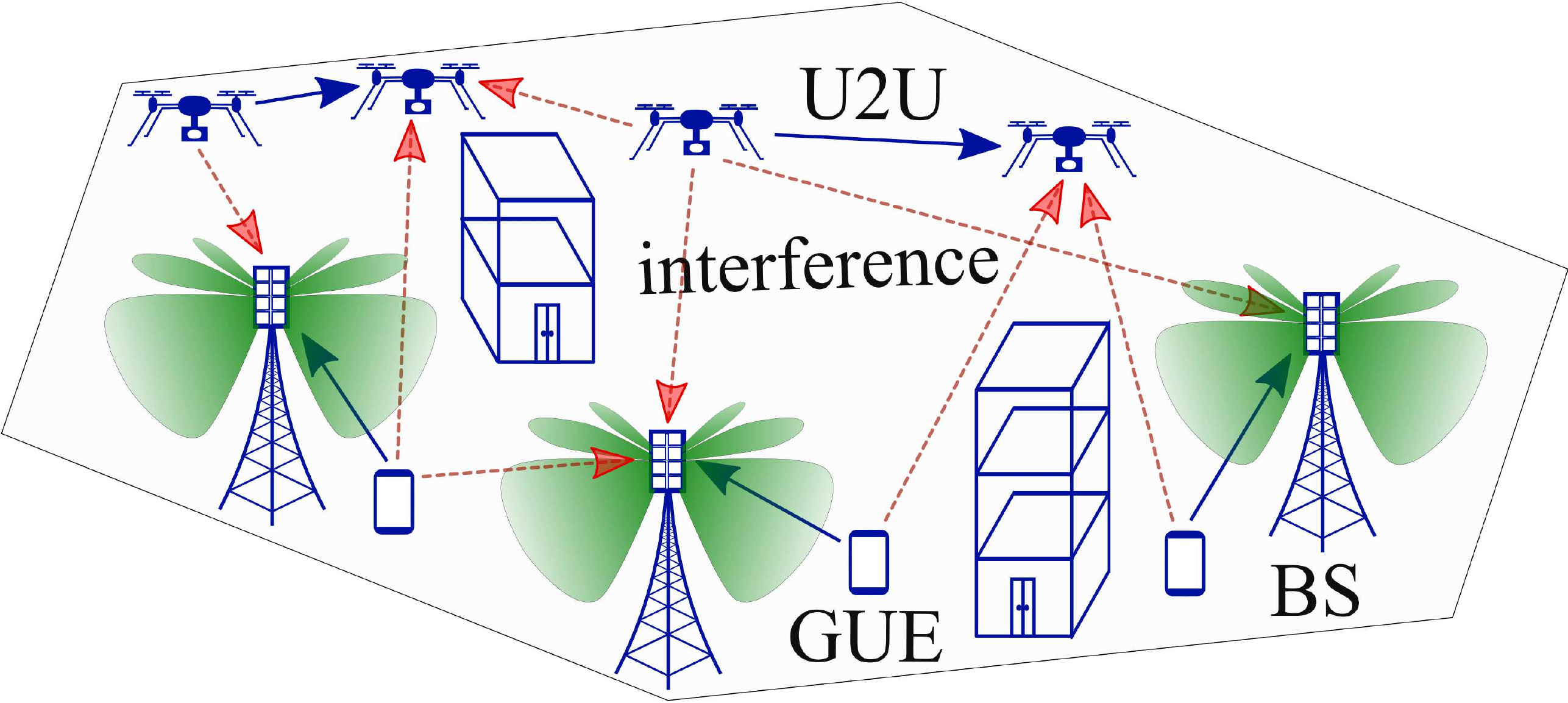}
	\caption{Illustration of U2U communications sharing spectrum with the uplink of a cellular network. Blue solid arrows indicate communication links, whereas red dashed arrow indicate interfering links.}
	\label{U2U_SystemModel}
\end{figure}
\subsubsection*{Antenna gain}
We assume that all GUEs and UAVs are equipped with a single omnidirectional antenna with unitary gain. On the other hand, we consider a realistic BS antenna radiation pattern to capture the effect of sidelobes, which is of particular importance in UAV-to-BS links \cite{azari2017coexistence,GerGarGal2018}. We assume that each BS is equipped with a vertical, $\mathrm{N}$-element uniform linear array (ULA), where each element has directivity
\begin{equation} \label{ElementGain}
g_E(\theta) = g_E^{\max} \sin^2\theta
\end{equation}
as a function of the zenith angle $\theta$. The total BS radiation pattern $g_b(\theta) = g_E(\theta)\cdot g_A(\theta)$ is obtained as the superposition of each element's radiation pattern $g_E(\theta)$ and by accounting for the array factor given by
\begin{equation}
g_A(\theta) = \frac{\sin^2\Big(N\pi (\cos\theta - \cos\tt)/2\Big)}{N\sin^2\Big(\pi (\cos\theta - \cos\tt)/2\Big)},	
\end{equation}
where $\tt$ denotes the electrical downtilt angle.
The total antenna gain $\gxy$ between a pair of nodes $\x$ and $\y$ is given by the product of their respective antenna gains.

\subsubsection*{Small-scale fading}
On a given PRB, $\sixy$ denotes the small-scale fading power between nodes $\x$ and $\y$. Given the different propagation features of ground-to-ground, air-to-air, and air-to-ground links, we adopt the general Nakagami-m small-scale fading model. As a result, the cumulative distribution function (CDF) of $\sixy$ is given by
\begin{equation} \label{FadingCDF}
F_{\sixy}(\omega) \triangleq \mathbb{P}[\sixy < \omega] \!=\! 1\!-\!\sum_{i=0}^{\mxy-1} \!\frac{(\mxy \omega)^i}{i!} e^{-\mxy \omega},\!
\end{equation}
where $\mxy \in \mathbb{Z}^{+}$ is the fading parameter, with LoS links typically exhibiting a larger value of $\mxy$ than NLoS links.

\subsubsection*{Power Control}
As per the 3GPP guidelines, we consider fractional power control for all nodes. Accordingly, the power transmitted by node $\x$ while communicating to node $\y$ is given by \cite{BarGalGar2018GC}
\begin{equation}
\Px = \min\left\{ \Pxmax, \rhox \cdot \zetaxy^{\epsx} \right\},
\label{eqn:power_control}
\end{equation}
where $\Pxmax$ is the maximum transmit power at node $\x$, $\rhox$ is a cell-specific parameter, $\epsx \in [0,1]$ is the fractional power control factor, and $\zetaxy = \tauxy /\gxy$ is the large-scale fading between nodes $\x$ and $\y$. The aim of (\ref{eqn:power_control}) is to compensate for a fraction $\epsx$ of the large-scale fading, up to a limit imposed by $\Pxmax$ \cite{3GPP36777}.

\section{Analytical Results}
\label{sec:analysis_underlay}

Our U2U (resp. GUE UL) performance analysis is conducted for a typical BS (resp. UAV) receiver located at the origin.
In what follows, uppercase and lowercase are employed to respectively denote random variables and their realizations, e.g., $\Ru$ and $\ru$. Throughout the derivations, we make use of the superscripts $\nu,\xi \in \{\mathrm{L},\mathrm{N}\}$ to denote LoS and NLoS conditions on a certain link.

\subsection{U2U Performance Analysis}
We now derive the U2U link coverage, i.e., the complementary CDF (CCDF) of the signal-to-interference-plus-noise ratio (SINR) experienced by a UAV.

\begin{Theorem} \label{U2Ucov_theorem}
	The U2U link coverage is given by
	\begin{align}
	\pcovuav \!= \!\!\!\!\! \sum_{\nu \in \{\mathrm{L},\mathrm{N}\}}\int_0^\infty \hspace{-0.2cm} f_{\Ru}^\nu(\ru) \!\!\sum_{i=0}^{\mathrm{m_{uu}^\nu}-1} \!\!(-1)^i \q_{\u,i}^\nu \cdot \!\D^i_{\s_\u} \!\left[\lapiu^\nu(\mathrm{s_u})\right] \! \mathrm{d}\ru,
	\label{eq:U2ULinkCoverage}
	\end{align}
	where $\D^i_{\s_\u}$ represents the i-th derivative with respect to $\s_\u$ and $\Ru$ is the typical U2U communication link distance. Also, by denoting the noise power with $\mathrm{N_0}$, we have
	\begin{equation}
	\begin{aligned}
	\q_{\u,i}^\nu \!\triangleq\! \frac{e^{-\mathrm{N_0}\mathrm{s_u}}}{i!}\sum_{j=i}^{\mathrm{m_{uu}^\nu}-1} \frac{\mathrm{N_0}^{j-i}\mathrm{s_u}^j}{(j-i)!},~~
	\!\!\mathrm{s_u} \!\triangleq\! \frac{\mathrm{m_{uu}^\nu} \t}{\pu^\nu(\ru) \zuu^\nu(\ru)^{-1}}.
	\end{aligned}
	\end{equation}
	In \eqref{eq:U2ULinkCoverage}, the interference is characterized by its Laplacian, which is obtained as $\lapiu^\nu(\mathrm{s_u}) = e^{\eta(\mathrm{s_u})}$ with
	\begin{equation}
	\eta(\mathrm{s_u}) \!=\! -2 \pi\!\left[ \lamu \!\!\!\!\sum_{\xi \in \{\mathrm{L},\mathrm{N}\}}\!\!\mathcal{I}_\mathrm{uu}^\mathrm{\xi}(\mathrm{s_u})  \!+\! \lamb \!\!\!\!\sum_{\xi \in \{\mathrm{L},\mathrm{N}\}}\mathcal{I}_\mathrm{cu}^\mathrm{\xi}(\mathrm{s_u}) \right]\!,
	\end{equation}
		where for $\xi \in \{\mathrm{L},\mathrm{N}\}$
	\begin{align}
	\mathcal{I}_\mathrm{xy}^\xi &= \int_0^\infty \!\hspace{-0.2cm} f_{R_\mathrm{x}}^\mathrm{L}(x)\sum_{i = 1}^\infty \left[\prxy^{\xi}(\r_{i-1})\!-\!\prxy^{\xi}(\r_{i})\right] \underbrace{\Psi_\mathrm{xy}^\xi\left(\mathrm{s},\r_{i}\right)}_\text{at $P_\x = P_\mathrm{x}^\mathrm{L}$}  \!\mathrm{d}x\nonumber
	\\ 
	&\!\!\!\!\!\!\!\!+\! \int_0^\infty \!\hspace{-0.2cm} f_{R_\mathrm{x}}^\mathrm{N}(x) \sum_{i = 1}^\infty \!\left[\prxy^{\xi}(\r_{i-1})\!-\!\prxy^{\xi}(\r_{i})\right] \!\underbrace{\Psi_\mathrm{xy}^\xi\left(\mathrm{s},\r_{i}\right)}_\text{at $P_\x = P_\mathrm{x}^\mathrm{N}$}  \!\mathrm{d}x.\label{Ixyxi}
	\end{align}
	
	In \eqref{Ixyxi}, $\prxy^{\xi}(\r_0) \triangleq 0$, and
	\begin{equation}
	\begin{aligned}
	\Psi_\mathrm{xy}^\xi(\mathrm{s},\r) &\triangleq \frac{\r^2+\mathrm{h_{xy}^2}}{2}\left[1-\left(\frac{\m}{\m+\mu}\right)^{\m}\right] \\
	&\!\!\!\!\!\!\!- \mathcal{K}(s)
	\,{_2}F_1\left(1+\m,1-\beta;2-\beta;-\frac{\mu}{\m}\right),
	\end{aligned}
	\end{equation}
	where ${_2}F_1(\cdot)$ is the Gauss hypergeometric function, $\m = \mathrm{m_{xy}^\xi}$, $\beta = \frac{2}{\alpha_\mathrm{xy}^\xi}$, $\mathrm{s} = \mathrm{s_y}\frac{\mathrm{g_{xy}}}{\hat{\tau}_{\mathrm{xy}}^\xi}$, and
	\begin{align}
	\mu(\s) \!\triangleq\! \frac{\mathrm{s} P_\mathrm{x}}{(\r^2\!+\!\mathrm{h_{xy}^2})^{1/\beta}},\!\!~\mathcal{K}(s) \!\triangleq\! \frac{\mathrm{s} P_\mathrm{x}}{2(1\!-\!\beta)\,(\r^2\!+\!\mathrm{h_{xy}^2})^{1/\beta-1}}.
	\end{align}
\end{Theorem}

\begin{proof}
	See Appendix \ref{U2Ucov_proof}.
\end{proof}

\subsection{GUE UL Performance Analysis}
We now obtain the GUE UL coverage, i.e., the CCDF of the UL SINR experienced by a GUE in the presence of U2U communications sharing the same spectrum.

\begin{Theorem} \label{C2Bcov_theorem}
	The GUE UL coverage is given by	
	\begin{align}
	\pcovc &\!= \!\!\!\!\! \sum_{\nu \in \{\mathrm{L},\mathrm{N}\}}\int_0^\infty \!\!\! f_{R_\c}^\nu(\rc) \!\! \sum_{i=0}^{\mathrm{m_{cb}^\nu}-1} \!\!(-1)^i \q_{\c,i}^\nu \cdot \D^i_{\s_\c} \!\left[ \lapic^\nu(\s_\c)\right]  \mathrm{d}\rc,
	\label{eq:GUEULCoverage}
	\end{align}
	where $\Rc$ is the GUE communication link distance to the typical BS and
	\begin{align}
	\q_{\c,i}^\nu &\triangleq \frac{e^{-\s_\c\mathrm{N_0}}}{i!}\sum_{j=i}^{\mathrm{m_{cb}^\nu}-1} \frac{\mathrm{N_0}^{j-i}\s_\c^j}{(j-i)!},~~	\s_\c \triangleq \frac{\mathrm{m_{cb}^\nu} \t}{\pcr \zcb^\nu(\rc)^{-1+\ec}}.
	\end{align}
	 In \eqref{eq:GUEULCoverage}, the interference is characterized by its Laplacian, which is obtained as
	\begin{align}
	\lapic &= e^{ -2 \pi \lamu \sum_{\xi \in \{\mathrm{L},\mathrm{N}\}}\mathcal{I}_\mathrm{ug}^\mathrm{\xi} } \cdot e^{ -(2 \pi \lamb)^2 \sum_{\xi \in \{\mathrm{L},\mathrm{N}\}}\mathcal{I}_\mathrm{gg}^\mathrm{\xi} },
	\end{align}
	where $\mathcal{I}_\mathrm{ug}^\mathrm{\xi}$ is 
	\begin{align} \nonumber
	\mathcal{I}_\mathrm{ug}^\xi \!&= \!\!\! \int_0^\infty \!\hspace{-0.2cm} f_{\Ru}^\L(x) \!\sum_{i = 1}^\infty \pub^{\xi}(\r_i) \Big(\!\underbrace{\Psi_\mathrm{ub}^\xi\left(\mathrm{s},\r_{i+1}\right) \!-\! \Psi_\mathrm{ub}^\xi\left(\mathrm{s},\r_{i}\right)}_\text{at $P_\mathrm{u} = P_\mathrm{u}^\mathrm{L}$}\!\Big)  \mathrm{d}x
	\\
	&\!\!+\!\! \int_0^\infty \!\hspace{-0.2cm} f_{\Ru}^\mathrm{N}(x)\!\sum_{i = 1}^\infty \pub^{\xi}(\r_i) \Big(\!\underbrace{\Psi_\mathrm{ub}^\xi\left(\mathrm{s},\r_{i+1}\right) \!-\! \Psi_\mathrm{ub}^\xi\left(\mathrm{s},\r_{i}\right)}_\text{at $P_\mathrm{u} = P_\mathrm{u}^\mathrm{N}$}\!\Big)  \,\mathrm{d}x,
	\end{align}
	with $\mathrm{s} = \mathrm{s_{c}}\frac{\mathrm{g_{ub}(\r_i)}}{\hat{\tau}_{\mathrm{ub}}^\xi}$, whereas
	\begin{equation}
	\begin{aligned} \label{Iccxi}
	\mathcal{I}_\mathrm{gg}^\xi &\!\!=\!\! \int_0^\infty \!\prcb^\L(x) x e^{-\lamb \pi x^2} \\
	&\!\!\!\!\!\!\!\times \sum_{i = j(x)}^\mathrm{\infty} \!\!\prcb^{\xi}(\r_i)\, \!\Big(\!\underbrace{\Psi_\mathrm{gb}^\xi\left(\mathrm{s},\r_{i+1}\right) \!-\! \Psi_\mathrm{gb}^\xi\left(\mathrm{s},\r_{i}\right)}_{\text{at $P_{\g} = P_{\g}^\mathrm{L}$}} \!\Big) \mathrm{d}x \\
	&\!\!\!\!\!\!\!+ \int_0^\infty \prcb^\N(x) x e^{-\lamb \pi x^2} \\
	&\!\!\!\!\!\!\!\times \sum_{i = j(x)}^\mathrm{\infty} \prcb^{\xi}(\r_i)\, \Big(\underbrace{\Psi_\mathrm{gb}^\xi\left(\mathrm{s},\r_{i+1}\right) - \Psi_\mathrm{gb}^\xi\left(\mathrm{s},\r_{i}\right)}_{\text{at $P_\mathrm{g} = P_\mathrm{g}^\mathrm{N}$}} \Big) \mathrm{d}x,
	\end{aligned}
	\end{equation}
	with $\mathrm{s} = \mathrm{s_{g}}\frac{g_\mathrm{{gb}(\r_i)}}{\hat{\tau}_{\mathrm{gb}}^\xi}$. In \eqref{Iccxi}, we replace $r_{j(x)} = x$ where $j(x) = \lfloor{\frac{x\sqrt{\a_1\a_2}}{1000}\rfloor}+1$. 
\end{Theorem}

\begin{proof}
	See Appendix \ref{C2Bcov_proof}.
\end{proof}
\section{Numerical Results and Discussion}
\label{sec:numerical}

\begin{table}
	\centering
	\caption{System Parameters}
	\label{table:parameters}
	\def\arraystretch{1.2}
	\begin{tabulary}{\columnwidth}{ |p{2.2cm} | p{5.2cm} | }
		\hline
		\textbf{Deployment} 			&  \\ \hline  
		BS distribution		& PPP with $\lamb = 5$~/~Km$^2$\\ \hline
		GUE distribution 				& One active GUE per cell, $\hc=1.5$~m \\ \hline
		UAV distribution 				& $\lamu = 1$~/~Km$^2$, $\sigma_\mathrm{u} = 100\,$m, $\hc$=100~m \cite{3GPP36777} \\ \hline\hline
		\textbf{Channel model} 			&  \\ \hline
		\multirow{2}{*}{Ref. path loss [dB]}		&  $\hat{\tau}_\mathrm{cb}^\L = 28+20\log_{10}(f_c)$ \enspace ($f_c$ in GHz) \\ \cline{2-2}
		& $\hat{\tau}_\mathrm{gb}^\N = 13.54+20\log_{10}(f_c)$ \\ \cline{2-2}
		& $\hat{\tau}_\mathrm{ub}^\L = 28+20\log_{10}(f_c)$ \\ \cline{2-2}
		& $\hat{\tau}_\mathrm{ub}^\N = -17.5+20\log{10}(40\pi f_c/3)$ \\ \cline{2-2}
		& $\hat{\tau}_\mathrm{gu}^\L = 30.9+20\log_{10}(f_c)$  \\ \cline{2-2}
		&  $\hat{\tau}_\mathrm{gu}^\N = 32.4+20\log_{10}(f_c)$  \\ \cline{2-2}
		& $\hat{\tau}_\mathrm{uu}^\L = 28+20\log_{10}(f_c)$\\ \cline{2-2}
		& $\hat{\tau}_\mathrm{uu}^\N = -17.5+20\log{10}(40\pi f_c/3)$  \\ \hline
		\multirow{2}{*}{Path loss exponent}		&  $\alcb = 2.2,~~~\ancb = 3.9$ \\ \cline{2-2}
		& $\aubl = 2.2,~~~\aubn = 4.6-0.7\log_{10}(\hu)$ \\ \cline{2-2}
		& $\acul = 2.225-0.05\log_{10}(\hu)$ \\ & $\acun = 4.32-0.76\log_{10}(\hu)$ \\ \cline{2-2}
		& $\auul = 2.2,~~~\auun = 4.6-0.7\log_{10}(\hu)$\\ \hline
		Small-scale fading  & Rayleigh\ifx\[$\else\tablefootnote{
		After deriving analytical results under Nakagami-m small-scale fading, we now consider Rayleigh as a special case. This has been shown not to change the qualitative performance trends \cite{azari2017coexistence,LopDinLi2018GC}.}\fi, i.e.,  $\mxy^\xi = 1$ \\ \hline
		Prob. of LoS & ITU model as per \eqref{PrLoS} \\ \hline
		Thermal noise 				& -174 dBm/Hz with 7~dB noise figure \cite{3GPP36777}\\ \hline \hline
		\textbf{PHY} 			&  \\ \hline
		\multirow{2}{*}{Spectrum}		& Carrier frequency: 2~GHz \cite{3GPP36777} \\ \cline{2-2}
		& System bandwidth: 10 MHz with 50 PRBs\cite{3GPP36777}\\ \hline
		BS antennas 		& See \eqref{ElementGain} for elements gain 
		\\ \hline
		BS array configuration 		& Height: $25$~m, downtilt: $102^{\circ}$, $8\times 1$ vertical array, 1 RF chain, element spacing: $0.5~\lambda$ \cite{3GPP36777}\\ \hline
		Power control		& Fractional power control based on GUE-to-BS (resp. U2U) large-scale fading for GUEs (resp. UAVs), with $\ec = \eu = 0.6$, $\pcr = \pur = -58$~dBm, and $P^{\textrm{max}}_\c = P^{\textrm{max}}_\u = 24$~dBm \cite{BarGalGar2018GC}\\ \hline
		GUE/UAV antenna 		& Omnidirectional with 0~dBi gain \cite{3GPP36777} \\ \hline
		\end{tabulary}
\end{table}


We now provide numerical results to 
evaluate the performance of U2U and GUE UL communications sharing the same spectrum. Specifically, we concentrate on characterizing the impact that the UAV altitude, the UAV power control, and the U2U distance have on the U2U and GUE UL links. Unless otherwise specified, the system parameters used in this section are provided in Table~\ref{table:parameters}.

We model the U2U link distance $\Ru$ via a truncated Rayleigh distribution with probability density function (PDF)
\begin{align}
f_{\Ru}(\ru) = \frac{\ru e^{-\r_\u^2/(2\sigma_\mathrm{u}^2)}}{\sigma_\mathrm{u}^2\left(1-\e^{-r_\mathrm{M}^2/(2\sigma_\mathrm{u}^2)}\right)}\cdot\mathds{1}(\ru<r_\mathrm{M}),
\end{align}
where $\r_\mathrm{M}$ is the maximum U2U link distance, $\mathds{1}(\cdot)$ is the indicator function, and $\sigma_\mathrm{u}$ is the Rayleigh scale parameter, related to the mean distance $\bar{R}_\mathrm{u}$ through $\sigma_\mathrm{u} = \sqrt{\frac{2}{\pi}}\bar{R}_\mathrm{u}$.

As for the probability of LoS between any pair of nodes x and y, we employ the well known ITU model \cite{ITU2012, azari2018cellular}:
\begin{equation}
\prxy^\L(r) \!=\! \!\!\!\!\!\!\! \prod_{j=0}^{\lfloor{\frac{r\sqrt{\a_1\a_2}}{1000}-1\rfloor}} \!\!\left[1\!-\!\exp\!\left(\!-\frac{\left[\!\hx\!-\!\frac{(j+0.5)(\hx-\hy)}{\mathrm{k}+1} \!\right]^2}{2 \a_3^2}\right) \!\right]\!,\!\!
\label{PrLoS}
\end{equation}
where $\{\a_1, \a_2, \a_3\}$ are environment-dependent parameters set to $\{0.3, 500, 20\}$ to model an urban scenario \cite{ITU2012}.

Fig.~\ref{Pcov_T_hu} shows the CCDF of the SINR per PRB experienced by: (i) U2U links, (ii) the UL of GUEs in the presence of U2U links, and (iii) the UL of GUEs without any U2U links. For (i) and (ii), we consider two UAV heights, namely 50~m and 150~m. 
Fig.~\ref{Pcov_T_hu} also allows to make the following observations:
\begin{itemize}[leftmargin=*]
\item U2U communications degrade the UL performance of GUEs. However, such performance loss amounts to less than 3~dB in median, since (i) BSs perceive interfering UAVs through their antenna sidelobes, and (ii) UAVs generally transmit with low power due to the good U2U channel conditions.
\item The U2U performance degrades as UAVs fly higher, due to an increased UAV-to-UAV and GUE-to-UAV interference. The former is caused by a higher probability of LoS between a receiving UAV and interfering UAVs. The latter is caused by a higher probability of LoS between a receiving UAV and interfering GUEs, whose effect outweighs having larger GUE-UAV distances.
\item The GUE UL performance also degrades as UAVs fly higher. However, this degradation is less significant than that experienced by the U2U links, since interference generated by GUEs in other cells is dominant.
\end{itemize}

\begin{figure}
	\centering
	\includegraphics[width=\figwidth]{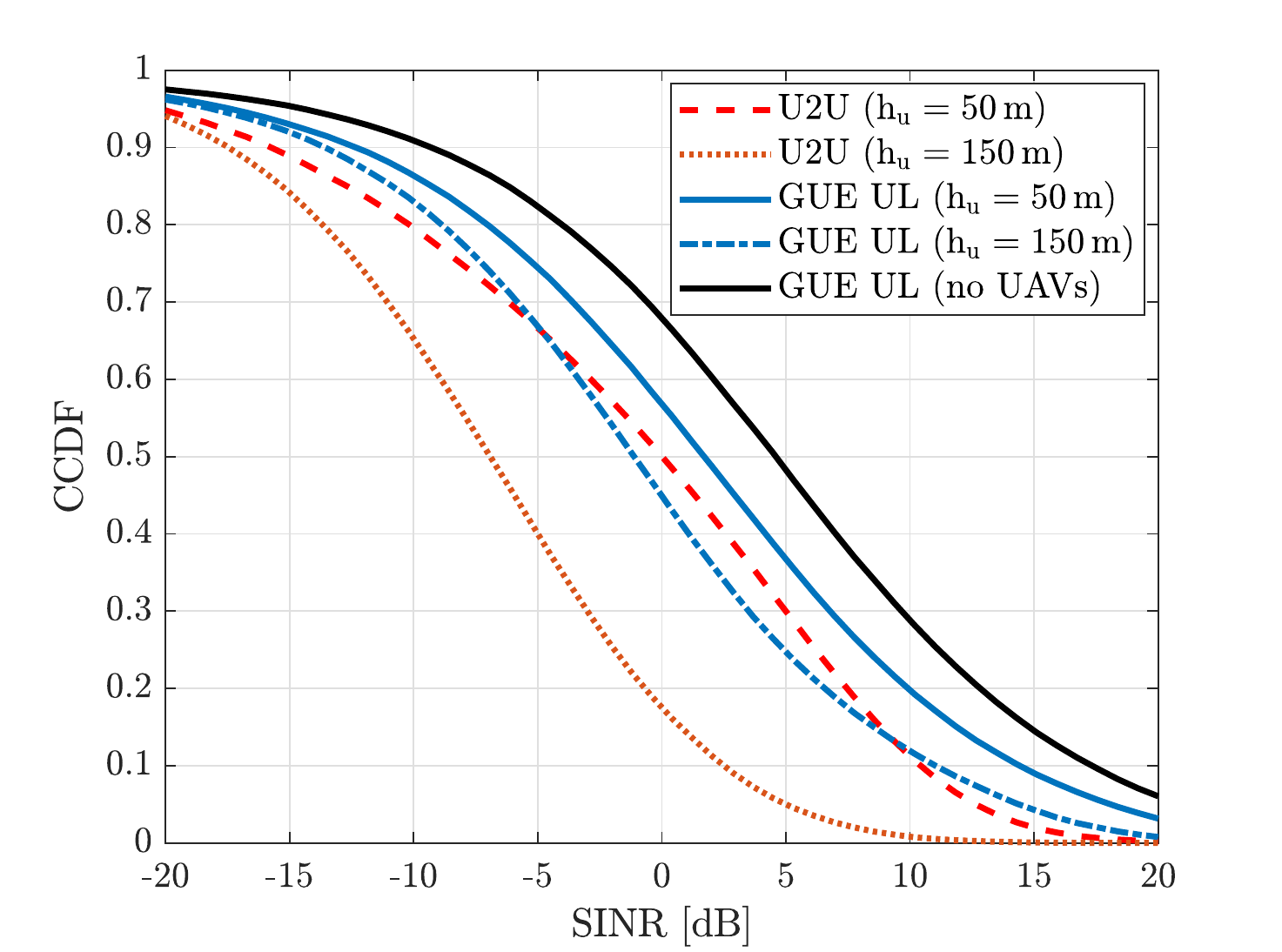}
	\caption{CCDF of the SINR per PRB experienced by: (i) U2U links, (ii) GUE UL in the presence of U2U links, and (iii) GUE UL without U2U links, for $h_u = \lbrace 50, 150 \rbrace$~m.}
	\label{Pcov_T_hu}
\end{figure}

Fig.~\ref{SigInt_epsU_UAVandGUE} illustrates (i) the mean useful received power, (ii) the mean interference power received from GUEs, and (iii) the mean interference power received from UAVs, for both U2U and GUE UL links. These metrics are plotted as a function of the fractional power control factor $\eu$ employed by UAVs. We may observe the following:
\begin{itemize}[leftmargin=*]
\item The UAV power control policy has a significant impact on the performance of both U2U and UL GUE links.
\item In the scenario under consideration, the mean interference perceived by GUEs is dominated by the GUE-generated transmissions from other cells for $\eu<0.6$, where such interference also remains small compared to the mean useful received power. The interference generated by UAVs dominates for larger values of $\eu$, and it saturates for $\eu > 0.9$, since almost all UAVs transmit with their maximum allowed power.
\item The mean interference perceived by UAVs is dominated by the GUE-generated transmissions for $\eu<0.7$, where such interference is also relatively large compared to the mean useful received power. For larger values of $\eu$, the UAV-to-UAV interference becomes dominant and keep growing alongside the useful signal up to $\eu=0.9$, when almost all UAVs operate at maximum power.
\end{itemize}

\begin{figure}
	\centering
	\includegraphics[width=\figwidth]{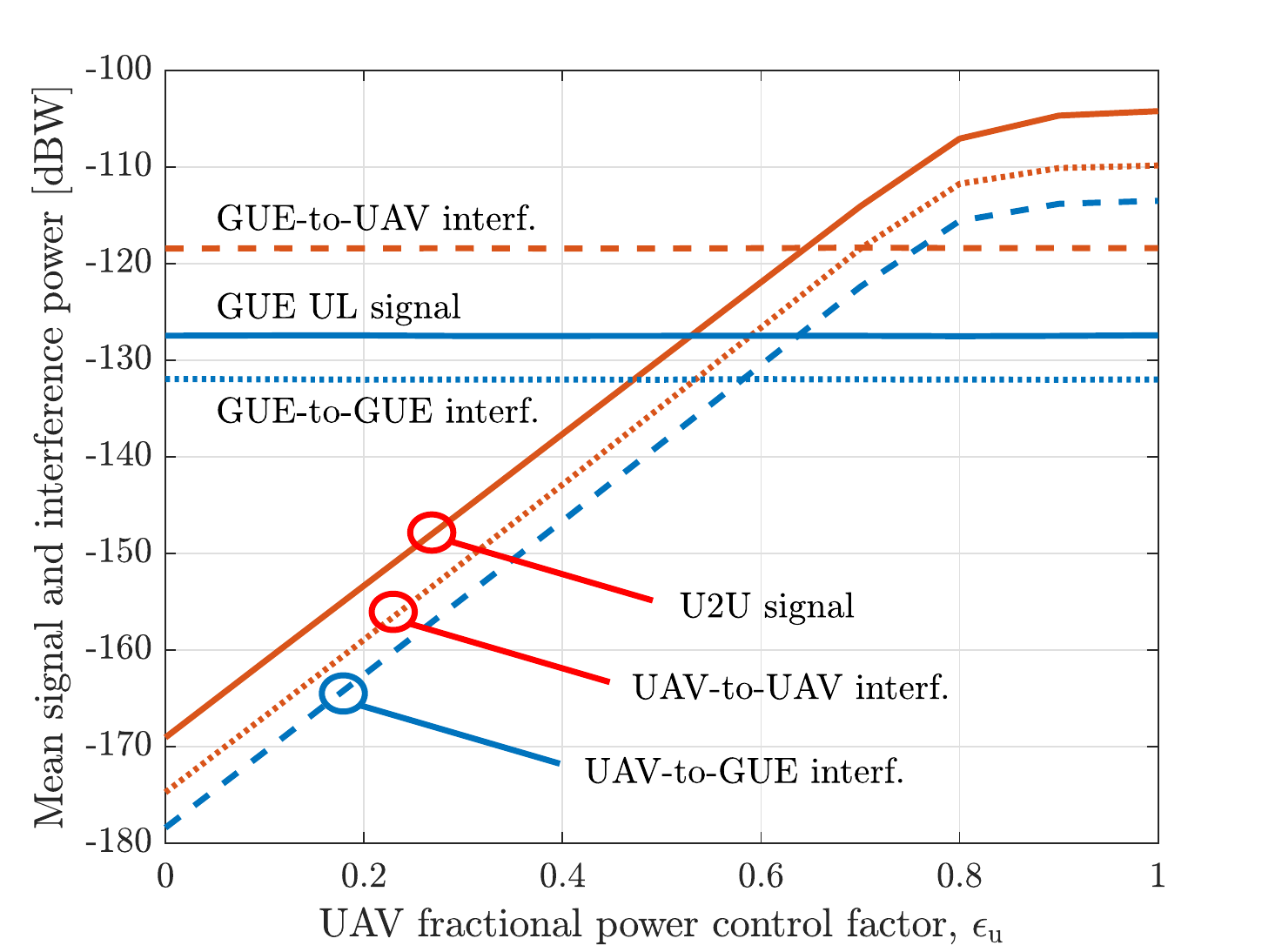}
	\caption{Mean values of the (i) received signal power, (ii) interference generated by GUEs, and (iii) interference generated by UAVs, as a function of the UAV's fractional power control factor $\eu$, and for both U2U links and GUE UL.}
	\label{SigInt_epsU_UAVandGUE}
\end{figure}

Fig.~\ref{Pcov_epsU_T} shows the probability of experiencing SINRs per PRB larger than -5~dB for both U2U and UL GUE links as a function of $\eu$. We also consider three different values for the U2U distance parameter $\sigmau$, namely 50~m, 100~m, and 150~m, corresponding to mean U2U distances $\bar{R}_\mathrm{u}$ of 63~m, 125~m, and 188~m, respectively. Fig.~\ref{Pcov_epsU_T} allows us to draw the following conclusions:
\begin{itemize}[leftmargin=*]
\item There exists an inherent tradeoff between the performance of U2U and GUE UL, whereby increasing $\eu$ improves the former at the expense of the latter: 
\begin{itemize}[leftmargin=*]
\item For $0 < \eu < 0.4$, the U2U performance is deficient, since UAVs use a very low transmission power. In this range, the UL GUE performance is approximately constant, since the GUE-generated interference is dominant.
\item For $0.4 < \eu < 0.9$, the U2U performance increases at the expense of the UL GUE links.
\item For $\eu > 0.9$, the U2U performance saturates and that of the GUEs stabilizes, since almost all aerial devices reach their maximum transmit power.
\end{itemize}
\item Smaller U2U link distances---for fixed UAV density---correspond to a better U2U performance for all values of $\eu$. This is because (i) UAVs perceive larger received signal powers for decreasing $\sigmau$, since the path loss of the U2U links diminishes faster than the UAV transmit power when $\sigmau$ lessens, and (ii) the reduced UAV-to-UAV interference due to the smaller transmission power employed by UAVs.
\item The UL GUE links also benefit from smaller U2U link distances when $\eu > 0.4$, since UAVs lower their transmit power and therefore reduce the UAV-to-BS interference.
\end{itemize}

\begin{figure}
	\centering
	\includegraphics[width=\figwidth]{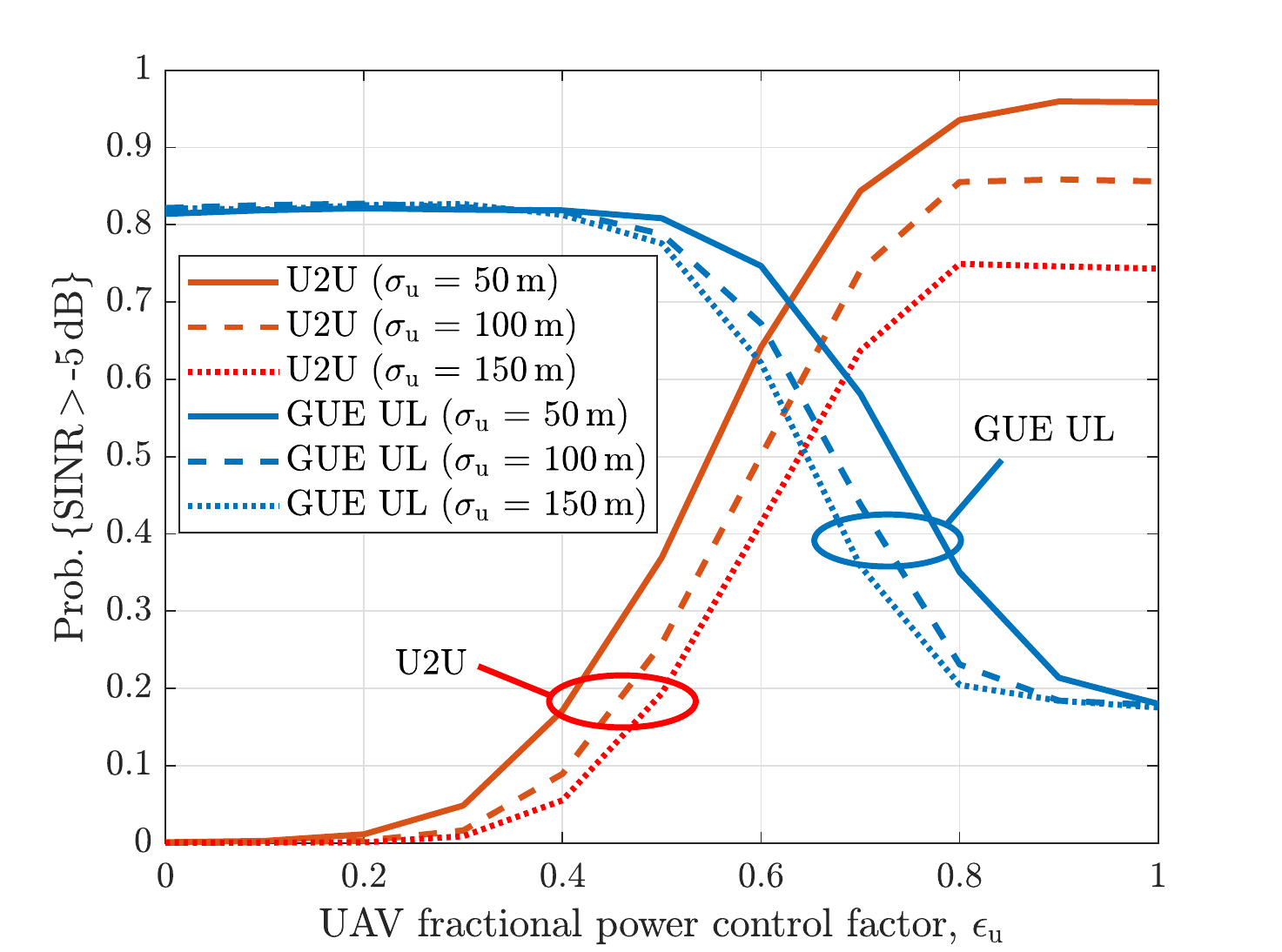}
	\caption{Probability of having SINRs $>-5$~dB for U2U and GUE UL links vs. the UAV's fractional power control factor $\eu$, and for $\sigmau = \lbrace 50, 100, 150 \rbrace$.}
	\label{Pcov_epsU_T}
\end{figure}
\section{Conclusion}
\label{sec:conclusion}

In this article, we have provided an analytical framework to study U2U communications underlayed with the UL of a cellular network. When considering a realistic channel model, antenna pattern, and power control policy, we demonstrated that communications between pairs of close-by UAVs have a limited effect on the GUE UL, since the strong U2U channel gains allow UAVs to lower their transmit power. Our results also showed that both the U2U and GUE UL SINRs diminish as UAVs fly higher, since aerial equipment encounters LoS propagation conditions with a larger number of nodes, which leads to an overall interference growth. We also demonstrated how the UAV power control policy serves to find a performance trade-off between U2U and UL GUE communications.

%



\begin{appendix}

\subsection{Sketch of Proof of Theorem \ref{U2Ucov_theorem}} \label{U2Ucov_proof}

	To obtain the U2U link coverage we can write
	
	\begin{align}
	\pcovuav
	& = \sum_{\nu \in \{\mathrm{L},\mathrm{N}\}}\int_0^{\rMuu}  \mathcal{C}_{\mathrm{u}|\Ru}^\nu(\ru)\,f_{\Ru}^\nu(\ru) \,\mathrm{d}\ru,
	\end{align}
	where
	\begin{align} \label{CuRu}
	\mathcal{C}_{\mathrm{u}|\Ru}^\nu(\ru) 
	&= \sum_{i=0}^{\mathrm{m_{uu}^\nu}-1} (-1)^i \q_{\u,i}^\nu \cdot \D^i_{\mathrm{s_u}} \left[ \lapiu^\nu(\mathrm{s_u}) \right],
	\end{align}
	obtained using the CDF of gamma-distributed small-scale fading in \eqref{FadingCDF}. As for the interference, one can write
	\begin{align} \label{lapiu}
	\lapiu^\nu(\mathrm{s_u}) \!=\! \lapicul^\nu(\mathrm{s_u}) \!\cdot\! \lapicun^\nu(\mathrm{s_u}) \!\cdot\! \lapiuul^\nu(\mathrm{s_u}) \!\cdot\! \lapiuun^\nu(\mathrm{s_u}),
	\end{align}
	where $I_\mathrm{xy}^\xi$ is the interference imposed by nodes x of condition $\xi$ on y. Each term in \eqref{lapiu} can be characterized as follows:
	\begin{align}
	\lapixyxi^\nu &= e^{ -2 \pi \lambda_\mathrm{x} \mathcal{I}_\mathrm{xy}^\xi };~~\xi \in \{\L,\N\}
	\end{align}
	where $\lambda_\c=\lamb$ accounts for the density of active users, and
	\begin{align} \label{IxyxiCal}
	&\mathcal{I}_\mathrm{xy}^\xi \!=\!
	\sum_{i = 1}^\infty \prxy^\xi(\r_i) \mathbb{E}_{\px,\sixy^\xi} \hspace{-0.15cm} \left[ \int_{\r_{i}}^{\r_{i+1}} \hspace{-0.2cm} \Big(\!1 \!-\! e^{-\s \px d_\mathrm{xy}^{-\axy^\xi} \sixy^\xi} \!\Big) r \mathrm{d}r \right],
	\end{align}
	with $\s = \mathrm{s_y} \frac{\mathrm{g_{xy}}(\r_i)}{\hat{\tau}_{\mathrm{xy}}^\xi} $.
	In the following we calculate the integral term in \eqref{IxyxiCal} by considering a change of variable as $\omega = \s \px \,d_\mathrm{xy}^{-\axy^\xi} \, \sixy^\xi$. Therefore we can rewrite 
	\begin{equation}
	\begin{aligned}
	\int_{\r_{i}}^{\r_{i+1}}  &\Big(1- e^{-\s \px \,d_\mathrm{xy}^{-\axy^\xi} \, \sixy^\xi} \Big)\, r \,\mathrm{d}r \\
	&\!\!\!\!\!\!= \frac{(\s \px \sixy^\xi)^{\bxy^\xi}}{\axy^\xi} \int_{ \omega_2}^{ \omega_1} \omega^{-1-\bxy^\xi} (1-e^{-\omega}) \d\omega,
	\end{aligned}
	\end{equation}
	where $\bxy^\xi \triangleq 2/\axy^\xi$, $\omega_1 =  \mu_1 \sixy^\xi$, $\omega_2 =  \mu_2 \sixy^\xi$, and
	\begin{align}
	\mu_1 &\triangleq \frac{\s \px}{(\r_i^2+\mathrm{h_{xy}^2})^{\axy^\xi/2}},~~~	\mu_2 \triangleq \frac{\s \px}{(\r_{i+1}^2+\mathrm{h_{xy}^2})^{\axy^\xi/2}}.
	\end{align}
	We also have
	\begin{align} \nonumber
	\int_{ \omega_2}^{ \omega_1} &\omega^{-1-\bxy^\xi} (1-e^{-\omega}) \d\omega = \frac{\axy^\xi}{2} \Big[ \omega_2^{-\bxy^\xi} (1-e^{-\omega_2}) \\
	&- \omega_1^{-\bxy^\xi} (1-e^{-\omega_1})  + \int_{\omega_2 }^{\omega_1} y^{-\bxy^\xi} e^{-\omega}  \,\mathrm{d}\omega \Big],
	\end{align}
	where integration by parts is applied, and
	\begin{align}
	\int_{\omega_2 }^{\omega_1} \!\!\omega^{-\bxy^\xi} e^{-\omega}  \,\mathrm{d}\omega \!=\! \gamma\!\left(1\!-\!\bxy^\xi,\omega_1\right) \!-\! \gamma\left(1\!-\!\bxy^\xi,\omega_2\right),
	\end{align}
	where we used the definition of incomplete gamma function.
	It follows that
	\begin{equation}
	\begin{aligned} \label{integralStep}
	&\int_{\r_{i}}^{\r_{i+1}}  \Big(1- e^{-\s \px \,d_\mathrm{xy}^{-\axy^\xi} \, \sixy^\xi} \Big)\, r \,\mathrm{d}r \\ 
	& = \frac{\r_{i+1}^2+\mathrm{h_{xy}^2}}{2}(1-e^{-\mu_2\sixy^\xi}) - \frac{\r_{i}^2+\mathrm{h_{xy}^2}}{2}(1-e^{-\mu_1 \sixy^\xi}) \\
	&\!+ \frac{(\s \px \sixy^\xi)^{\bxy^\xi}}{2} \Big[ \! \gamma\!\left(1\!-\!\bxy^\xi,\mu_2\sixy^\xi \right) \!-\!  \gamma\!\left(1\!-\!\bxy^\xi,\mu_1 \sixy^\xi \right) \!\Big]\!.
	\end{aligned}
	\end{equation}
	We note that for Nakagami-m fading $\psi$ with parameter $\mathrm{m}$ we have 
	\begin{align}
	\mathbb{E}_{\psi} \left[e^{-\mu \psi}\right] = \left(1+\frac{\mu}{\mathrm{m}}\right)^{-\mathrm{m}}
	\end{align}
	and
	\small
	\begin{equation}
	\begin{aligned}
	&\mathbb{E}_{\psi} \left[\psi^{\beta} \gamma \left(1-\beta,\mu \psi \right)\right]  \\
	\!\!\!\!\!\!\!&= \frac{\mathrm{m}^{1+\mathrm{m}} \mu^{1-\beta}}{(1-\beta)(\mathrm{m}+\mu)^{1+\mathrm{m}}}\,{_2}F_1\left(1,1\!+\!\mathrm{m};2\!-\!\beta;\frac{\mu}{\mu+\mathrm{m}}\right).
	\end{aligned}
	\end{equation}
	\normalsize
	Now through transformation properties of the hypergeometric function 
	we can write
	\begin{equation}
	\begin{aligned} \label{hyperFunc}
	&{_2}F_1\left(1,1+\m;2-\beta;\frac{\mu}{\mu+\m}\right) \\
	&= \left(\frac{\m}{\mu\!+\!\m}\right)^{\!-1-\m} \!\!\!\!\!{_2}F_1\left(1\!+\!\m,1\!-\!\beta;2\!-\!\beta;-\frac{\mu}{\m}\right).
	\end{aligned}
	\end{equation}
	Consequently, by using \eqref{integralStep}--\eqref{hyperFunc} we have
	\begin{equation}
	\begin{aligned}
	\mathbb{E}_{\sixy^\xi} &\left[ \int_{\r_{i}}^{\r_{i+1}}  \Big(1- e^{-\s \px \,d_\mathrm{xy}^{-\axy^\xi} \, \sixy^\xi} \Big)\, r \,\mathrm{d}r \right] \\ &= \Psi_\mathrm{xy}^\xi\left(\s,\r_{i+1}\right) - \Psi_\mathrm{xy}^\xi\left(\s,\r_{i}\right).
	\end{aligned}
	\end{equation}
	Accordingly,
	\begin{align} \nonumber \label{IxyFinal}
	&\mathcal{I}_\mathrm{xy}^\xi 
	\!=\! \int_0^\infty  \!\!\!f_{R_\mathrm{x}}^\mathrm{L}(x)\!\sum_{i = 1}^\infty \prxy^{\xi} \Big(\!\underbrace{\Psi_\mathrm{xy}^\xi\left(\mathrm{s},\r_{i+1}\right) \!-\! \Psi_\mathrm{xy}^\xi\left(\mathrm{s},\r_{i}\right)}_\text{computed at $P_\mathrm{x}^\mathrm{L}$}\!\Big)  \mathrm{d}x\\
	&+ \int_0^\infty  f_{R_\mathrm{x}}^\mathrm{N}(x)\,\sum_{i = 1}^\infty \prxy^{\xi} \Big(\underbrace{\Psi_\mathrm{xy}^\xi\left(\mathrm{s},\r_{i+1}\right) - \Psi_\mathrm{xy}^\xi\left(\mathrm{s},\r_{i}\right)}_\text{computed at $P_\mathrm{x}^\mathrm{N}$}\Big)  \,\mathrm{d}x.
	\end{align}
	By noting that
	\begin{equation}
	\begin{aligned}
	\sum_{i = 1}^\infty \prxy^\xi(\r_i) &\Big(\Psi_\mathrm{xy}^\xi\left(\s,\r_{i+1}\right) - \Psi_\mathrm{xy}^\xi\left(\s,\r_{i}\right)\Big) \\
	&\!\!\!= \sum_{i = 1}^\infty \left[\prxy^{\xi}(\r_{i-1})-\prxy^{\xi}(\r_{i})\right] \Psi_\mathrm{xy}^\xi\left(\mathrm{s},\r_{i}\right),
	\end{aligned}
	\end{equation}
	the desired result is obtained.

\subsection{Sketch of Proof of Theorem \ref{C2Bcov_theorem}} \label{C2Bcov_proof}

Similar to U2U coverage analysis, we can write
\begin{align}
\pcovc &
= \sum_{\nu \in \{\mathrm{L},\mathrm{N}\}}\int_0^\infty  \mathcal{C}_{\c|\Rc}^\nu(\rc)\,f_{\Rc}^\nu(\rc) \,\mathrm{d}\rc,
\end{align}
where
\begin{align}
\mathcal{C}_{\c|\Rc}^\nu(\rc) 
&= \sum_{i=0}^{\mathrm{m_{cb}^\nu}-1} (-1)^i \q_{\c,i}^\nu \cdot \D_{\s_\c}^i\left[ \lapic^\nu(\s_\c) \right],
\end{align}
where the last equation is derived similarly to \eqref{CuRu}.

The aggregate interference can be derived as follows
\begin{align}
\lapic^\nu(\s_\c) = \lapiucl^\nu(\s_\c) \!\cdot\! \lapiucn^\nu(\s_\c) \!\cdot\! \lapiccl^\nu(\s_\c) \!\cdot\! \lapiccn^\nu(\s_\c),
\end{align}
where $\lapiucl$ and $\lapiucn$ are obtained similarly to Theorem \ref{U2Ucov_theorem} by using \eqref{IxyFinal}. To characterize GUEs interference one can write
\begin{align}
\mathcal{L}_{I_\mathrm{gg}^\xi} &\!=\! e^ {-2 \pi \int_0^\infty \lamci(r) \Big(1-\mathbb{E}_{\pc,\scb^\xi} \left[ e^{-\s_\c \pc \zcb^\xi(r)^{-1} \scb^\xi} \right]\Big) r \mathrm{d}r }\!\!\!\!.
\end{align}
Therefore $\mathcal{L}_{I_\mathrm{gg}^\xi} = e^{-(2\pi\lamb)^2 \,\mathcal{I}_\mathrm{gg}^\xi}$, where
\small
\begin{equation}
\begin{aligned}
&\mathcal{I}_\mathrm{gg}^\xi = \sum_{\nu \in \{\L,\N\}} \int_0^\infty \prcb^\xi(r) \times\\
&\! \int_0^r  \!\!\prcb^\nu(x) x e^{-\lamb \pi x^2} \! \Bigg(\!1\!-\!\mathbb{E}_{\scb^\xi} \Bigg[ e^{-\frac{\s_\c \pc^\nu(x)  \, \scb^\xi}{\zcb^\xi(r)}} \Bigg]\Bigg) \mathrm{d}x r \mathrm{d}r.
\end{aligned}
\end{equation}
\normalsize
We can write
\begin{equation}
\begin{aligned}
\int_0^\infty \prcb^\xi(r) &\int_0^r \prcb^\nu(x) x e^{-\lamb \pi x^2} \\ 
&\!\!\!\!\times \Bigg(1-\mathbb{E}_{\scb^\xi} \Bigg[ e^{-\frac{\s_\c \pc^\nu(x)  \, \scb^\xi}{\zcb^\xi(r)}} \Bigg]\Bigg)\,\mathrm{d}x \,  r \,\mathrm{d}r \\
\!\!\!&\hspace{-2cm}= \int_0^\infty \prcb^\nu(x) x e^{-\lamb \pi x^2} \int_x^\infty  \prcb^\xi(r) \\ 
&\!\!\!\!\times \Bigg(1-\mathbb{E}_{\scb^\xi} \Bigg[ e^{-\frac{\s_\c \pc^\nu(x) \, \scb^\xi}{\zcb^\xi(r)}} \Bigg]\Bigg)\,r \,\mathrm{d}r\mathrm{d}x.
\end{aligned}
\end{equation}
\normalsize
To conclude the proof, we derive the inner integral as follows
\begin{equation}
\begin{aligned}
&\int_x^\infty  \prcb^\xi(r) \Bigg(1-\mathbb{E}_{\scb^\xi} \Bigg[ e^{-\frac{\s_\c \pc^\nu(x) \, \scb^\xi}{\zcb^\xi(r)}} \Bigg]\Bigg)\,r \,\mathrm{d}r \\
&= \sum_{i = j(x)}^\mathrm{\infty} \prcb^\xi(\r_i)\, \Bigg(\underbrace{\Psi_\mathrm{gb}^\xi\left(\s,\r_{i+1}\right) - \Psi_\mathrm{gb}^\xi\left(\s,\r_{i}\right)}_{\text{at $\pc = \pc^\nu(x)$}} \Bigg),
\end{aligned}
\end{equation}
where $\mathrm{s} = \s_\c \frac{g_\mathrm{{gb}(\r_i)}}{\hat{\tau}_\mathrm{gb}^\xi}$. Note that we assume the BS antenna gain is invariant within the interval $[r_{i},r_{i+1}]$ so that $g_\mathrm{gb}(r) = g_\mathrm{gb}(\r_i)$ is a constant value. 

\end{appendix}
\ifCLASSOPTIONcaptionsoff
  \newpage
\fi
\bibliographystyle{IEEEtran}
\bibliography{Strings_Gio,Bib_Gio,Bib_Mahdi}
\end{document}